\documentclass[aps,pra,twocolumn,10pt]{revtex4-2}

\usepackage{float,graphicx,multirow,amsmath,
color,dsfont,amsthm,amssymb,comment}
\usepackage{physics}
\usepackage[T1]{fontenc}
\usepackage[normalem]{ulem}

\usepackage{xcolor}
\definecolor{maroon}{RGB}{100,20,20}
\definecolor{dblue}{RGB}{20,20,100}

\usepackage[colorlinks=true,
linkcolor=dblue,
citecolor=maroon,
urlcolor=maroon]{hyperref}

\newcommand{\ssout}[1]{}

\newtheorem{definition}{Definition}
\newtheorem{proposition}[definition]{Proposition}

\begin{document}

\title{Ancilla-Assisted Quantum Process Tomography using
Bound entangled states}

\author{Gurvir Singh}
\email{gurvirsingh@iisermohali.ac.in}

\affiliation{
Department of Physical Sciences, Indian Institute of Science
Education and Research (IISER) Mohali,
Sector 81 SAS Nagar, Manauli PO 140306, Punjab, India}

\begin{abstract}
Lu \textit{et al.} [Ann. Phys. (Berlin) \textbf{534}, 2100550 (2022)]
raised the question of whether bound entangled states can be useful for
ancilla-assisted quantum process tomography (AAQPT). We answer their question in the affirmative by constructing bound entangled states that are faithful and hence allow complete reconstruction of an unknown quantum channel. We also study the effect of local filtering, which can increase the computable cross norm and realignment of a state while destroying its faithfulness. We further examine the performance of bound entangled probes, taking into account the stability of channel reconstruction rather than faithfulness alone. We find that bound entangled states can provide better-conditioned reconstruction than any separable probe, despite their entanglement being nondistillable.
\end{abstract}

\maketitle

\section{Introduction}

One of the fundamental challenges in  quantum information is
the physical verification and characterization of quantum
states. The protocol to completely determine a quantum state is
called Quantum State Tomography~\cite{Nielsen2010}. Given the well known
duality between quantum states and quantum channels via the
Choi--Jamio\l{}kowski isomorphism, one can pose a similar
question about characterization of processes (equivalently,
quantum channels). This procedure is known as Quantum
Process Tomography
(QPT)~\cite{Chuang1997,Poyatos1997,D_Ariano2003,Altepeter2003}.

The techniques to implement QPT can be broadly
classified into direct and indirect methods. Indirect
techniques rely on quantum state tomography to extract
information about the quantum process, requiring an
additional step of applying an inversion map to the output
data. The two main indirect approaches are Standard Quantum
Process Tomography~\cite{Chuang1997,Poyatos1997} and
Ancilla-Assisted Quantum Process Tomography
(AAQPT)~\cite{D_Ariano2003,Altepeter2003}. A
direct technique to perform QPT, known as Direct-
Characterization of Quantum Dynamics, has also been
developed~\cite{Mohseni2006}.
Comparative Resource analysis of these techniques have
subsequently been studied~\cite{Mohseni2008}.

In the present paper, we focus specifically on AAQPT, an
indirect technique that inherently depends on the
invertibility of the realigned probe matrix and exploits the
Choi--Jamio\l{}kowski isomorphism, where maximally entangled
states between the system and ancilla are employed to carry
out QPT~\cite{D_Ariano2003}. The property of invertibility,
termed \emph{faithfulness}, plays a crucial role in the
success of QPT via this method. It turns out that one can
replace maximally entangled states with other states, as
long as they satisfy the faithfulness criterion and still
carry out QPT successfully. Recently,
an operational definition of \emph{faithfulness} has been
formulated for two-qubit systems~\cite{Bao2024}. A
particularly intriguing observation in this context is that
even certain separable mixed states, Werner states and
Isotropic states, which satisfy the faithfulness criterion,
are useful for performing AAQPT~\cite{Altepeter2003}.

The study of bound entangled states has been a long-standing
area of interest in quantum information theory. The
existence of bound entangled states revealed a surprising
and intriguing aspect of quantum systems, and a complete
understanding of bound entanglement remains elusive.
A bipartite state is called \emph{bound entangled} if it is entangled
and yet no maximally entangled pair can be distilled from it by local
operations and classical communication, even asymptotically. Every entangled
state with a positive partial transpose (PPT) is of this
type~\cite{Horodecki1998}, and all bound entangled states considered in this
work are PPT entangled.
Although bound entangled states are not distillable, for a
long time it was believed that the entanglement present in
them had little practical utility. Over time, however,
several applications of bound entanglement have been
discovered. These include superactivation of bound
entanglement ($i.e.$ using tensor products of bound
entangled states to generate distillable
states)~\cite{Shor2003}, positive key rates in quantum key
distribution~\cite{Chi2007,Mishra2020}, and applications in
quantum metrology~\cite{Toth2018,Pal2021}. More recently,
bound entangled states have also been shown to be useful in
Prepare-and-Measure scenarios~\cite{Roch2025}. These
developments have significantly broadened the role of bound
entanglement in quantum information processing.

A natural question in this context is whether bound
entangled states can also serve as useful resource states
for AAQPT. Recently, Lu \textit{et al.}~\cite{Lu2022}
demonstrated that not all entangled states are useful for
AAQPT and provided an explicit example of a bound entangled
state that fails to perform this task. This leads to an
important question: \emph{Are all bound entangled
states useless for AAQPT?}

In this work, we answer this question in the negative by
showing that there exist bound entangled states which are useful for AAQPT. Our results therefore provide
another application of bound entanglement, further expanding
its scope as a useful resource in quantum information
theory.

In practice, one does not need a maximally entangled state to perform AAQPT. Any faithful state can be used, and small imperfections in preparing such a state do not destroy its faithfulness. What is interesting here is that even bound entangled states, despite being nondistillable, can be useful for this purpose. In fact, we show that the $4\otimes4$ bound entangled state considered below gives a better-conditioned reconstruction than any separable probe. In quantum metrology and in quantum key distribution the gain
relative to a maximally entangled resource is small, and the significance of
those results lies in the fact that a bound entangled state still delivers a
quantum advantage over every separable one. A similar advantage is found here in AAQPT, where we show that a bound entangled state can provide better-conditioned channel reconstruction than any separable probe.

\begin{table}[b]
\centering
\renewcommand{\arraystretch}{1.15}
\begin{tabular}{ll}
\hline\hline
Abbreviation & Meaning \\
\hline
QPT & Quantum process tomography \\
AAQPT & Ancilla-assisted quantum process tomography \\
CJ & Choi--Jamio\l{}kowski \\
CCNR & Computable cross norm realignment \\
PPT & Positive partial transpose \\
LOCC & Local operations and classical communication \\
SLOCC & Stochastic LOCC \\
HTP & Hidden teleportation power \\
OSR & Operator Schmidt rank \\
SV & Singular value \\
SN & Schmidt number \\
\hline\hline
\end{tabular}
\caption{Abbreviations used throughout this work.}
\label{tab:abbrev}
\end{table}

The rest of the paper is structured as follows:
In Sec.~\ref{preliminaries}, we introduce the necessary
notions required to understand AAQPT. In
Sec.~\ref{main_sec}, we present the main result, where we
show that bound entangled states are useful for AAQPT.
In Sec.~\ref{sec:local_filtering_AAQPT}, we discuss the role
of local filtering operation in the context of AAQPT. In
Sec.~\ref{comparison_aaqpt},
we compare the performance of bound entangled states to well
known states in context of AAQPT. Finally, we provide a
conclusion in Sec.~\ref{conclusion}.

\section{CJ Isomorphism and CCNR
criterion}\label{preliminaries}
A quantum channel describes the evolution
of a quantum system and is mathematically represented by a
linear, trace-preserving and completely positive (CPTP) map
$\mathcal{E}$. For any input state $\rho$, the output
$\mathcal{E}(\rho)$ can be expressed in the Kraus
form~\cite{Kraus1983}:
\begin{equation}
    \mathcal{E}(\rho) = \sum_n K_n \rho K_n^\dagger,
\end{equation}
where the Kraus operators  $\{K_n\}$ satisfy the
completeness relation, $\sum_n K_n^\dag K_n=\mathbb{I}$.
A fundamental challenge in quantum information is the
identification of an unknown channel $\mathcal{E}$, a task
known as quantum process tomography (QPT).

Consider a system together with an ancilla system and a
joint state $\rho_{AB}$ of the system and ancilla, as shown
in Fig.~\ref{fig:aaqpt_schematic}. Suppose that the quantum
map $\mathcal{E}$ acts on the system while the ancilla
remains unchanged. The resulting evolution of the joint
state is then described by $(\mathcal{E}\otimes
\mathbb{I})(\rho_{AB})$, where $\mathbb{I}$ denotes the
identity map acting on the ancilla. As a special case,
choosing $\rho_{AB}$ to be the maximally entangled state
$|\Phi^+\rangle = \sum_{i=1}^d (1 / \sqrt{d}) \;
|i\rangle \otimes |i\rangle \in \mathcal{H} \otimes
\mathcal{H},$
the resulting state can be written as
\begin{equation}\label{CJiso}
S_{\mathcal{E}}
=
(\mathcal{E}\otimes
\mathbb{I})(|\Phi^+\rangle\langle\Phi^+|).
\end{equation}

Using the Choi--Jamio\l{}kowski
isomorphism~\cite{Jamiolkowski1972,Choi1975,Arrighi2004},
this establishes a one-to-one correspondence between quantum
channels and quantum states. The inverse relation allows one
to reconstruct the channel from the Choi matrix
$S_{\mathcal{E}}$ through
\begin{equation}
\mathcal{E}(\rho)=
d\,\Tr_2\left[
(\mathbb{I}\otimes\rho^T)
S_{\mathcal{E}}
\right],
\end{equation}
where $\Tr_2$ denotes the partial trace over the second
subsystem and $T$ represents transposition in the chosen
basis.

From the Choi--Jamio\l{}kowski isomorphism, it is clear that 
Eq.~(\ref{CJiso}) contains complete information about
the quantum channel $\mathcal{E}$. A natural question that
arises in this context is whether one can replace the
maximally entangled state with some non-maximally
entangled state and still obtain information about ${\mathcal
E}$. This question was addressed
in~\cite{D_Ariano2003}, where it was shown that a
non-maximally entangled state can be used if and only if the singular
values of $\check{R}(\rho_{_{AB}})$ are nonzero ($i.e.$
$\check{R}(\rho_{_{AB}})^{-1}$ exists), where
$\check{R}(\rho_{_{AB}})$ is defined as follows,
\begin{eqnarray}
\check{R}(\rho_{_{AB}}):=(\rho_{_{AB}}^{T_B}\mathbb{F})^{T_A
},\label{checkR}
\end{eqnarray}
with  $\mathbb{F}=\sum_{ij}|ij\rangle\langle ji|$ being the
swap operator. A bipartite state which satisfies this property
is called \emph{faithful}. Further, as shown by Lu \textit{et al.}~\cite{Lu2022},
the singular values of
$\check{R}(\rho_{AB})$ and
$\mathcal{R}(\rho_{AB})$ are identical, where
$\mathcal{R}$ denotes the realignment operation.
The realignment operator plays an important role in
entanglement detection and, as discussed above, is also
crucial in determining whether a given state can be used for
AAQPT. Since it will be central to our discussion in the
next section, we briefly review it here.

Consider a bipartite quantum state $\rho_{AB}$ acting on
$\mathcal{H}_A\otimes\mathcal{H}_B$.
Given computational bases for $\mathcal{H}_A$ and
$\mathcal{H}_B$, the state can be written as
\begin{equation}
\rho_{AB}
=
\sum_{i,j,k,l}
\rho_{ij,kl}
|i\rangle\langle j|
\otimes
|k\rangle\langle l|,
\end{equation}
where $\rho_{ij,kl}$ are the matrix elements.
The realignment operator $\mathcal{R}$ is then defined
through its action on $\rho_{AB}$ as~\cite{Chen2003}
\begin{equation}
\mathcal{R}(\rho_{AB})
=
\sum_{i,j,k,l}
\rho_{ij,kl}
|i\rangle\langle k|
\otimes
|j\rangle\langle l|.
\end{equation}

The realignment operation is closely related to the
operator Schmidt decomposition of a bipartite state.
Any bipartite density operator can be expressed as
\begin{equation}
\rho_{AB}
=
\sum_{k=1}^{r}s_k A_k\otimes B_k,
\end{equation}
where $s_k>0$, and the sets of operators $\{A_k\}$ and
$\{B_k\}$ are orthonormal with respect to the
Hilbert--Schmidt inner product,
\begin{equation}
\Tr(A_i^\dagger A_j)
=
\Tr(B_i^\dagger B_j)
=
\delta_{ij}.
\end{equation}
The number of nonzero terms $r$ is called the
\emph{operator Schmidt rank} (OSR) of the state.
It is equal to the rank of the realigned matrix
\cite{Altepeter2003,MullerHermesNechita2018},
\begin{equation}
\operatorname{OSR}(\rho_{AB})
=
\operatorname{rank}\mathcal{R}(\rho_{AB}).
\end{equation}
In fact, the coefficients $s_k$ are precisely the
nonzero singular values of $\mathcal{R}(\rho_{AB})$.

For a bipartite state in
$\mathbb{C}^d\otimes\mathbb{C}^d$, the maximal
operator Schmidt rank is $d^2$. A state is faithful
for AAQPT if and only if it has maximal operator
Schmidt rank~\cite{Altepeter2003,D_Ariano2003},
\begin{equation}
\operatorname{OSR}(\rho_{AB})=d^2,
\label{eq:osr_faithfulness}
\end{equation}
or equivalently, all the singular values of its
realigned matrix are nonzero.

The operator Schmidt rank should not be confused with
the Schmidt number $\operatorname{SN}(\rho)$, which
is the smallest integer $k$ such that $\rho$ can
be written as a mixture of pure states, each
having Schmidt rank at most $k$.
In particular, separable states have Schmidt
number one, but may still have maximal operator
Schmidt rank and therefore be faithful for AAQPT.

\begin{figure}[t]
    \centering
    \includegraphics[width=\columnwidth]{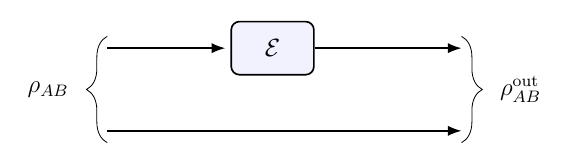}
    \caption{Schematic of ancilla-assisted quantum process
    tomography (AAQPT). The channel $\mathcal{E}$ acts on
    one subsystem of the bipartite state $\rho_{AB}$,
    while the ancilla remains unchanged. The resulting
    state
    $\rho^{\mathrm{out}}_{AB}
    =
    (\mathcal{E}\otimes\mathbb{I})(\rho_{AB})$
    contains information about the channel; for
    \textit{faithful} input states, this correspondence
    is one-to-one.}
    \label{fig:aaqpt_schematic}
\end{figure}

The same realigned matrix also appears in the
computable cross norm realignment (CCNR) criterion
for detecting quantum entanglement. While faithfulness
depends on the rank of $\mathcal{R}(\rho)$, the CCNR
criterion concerns the sum of its singular values.
Although introduced independently in different forms,
the various formulations of the CCNR criterion are
mathematically equivalent~\cite{Chen2003,Rudolph2005}.

The CCNR criterion states that for every bipartite
separable state $\rho_{\mathrm{sep}}$,
\begin{equation}
\|\mathcal{R}(\rho_{\mathrm{sep}})\|_{\mathrm{tr}}
\le 1,
\label{eq_CCNR}
\end{equation}
where
\[
\|X\|_{\mathrm{tr}}
=
\Tr\left(\sqrt{XX^\dagger}\right)
\]
denotes the trace norm (Schatten 1-norm).
Therefore, any violation,
$\|\mathcal{R}(\rho)\|_{\mathrm{tr}}>1$,
certifies that $\rho$ is entangled.

\section{Bound entangled states for AAQPT}\label{main_sec}

We now present the main result of this work by showing 
that a family of bound entangled states in $d \otimes d$
dimensions, with $d \ge 4$, can serve as faithful probes
for AAQPT.

To construct our example, we introduce three fundamental 
operators acting on the bipartite Hilbert space
$\mathbb{C}^d \otimes \mathbb{C}^d$. The identity
operator, the Flip (swap) operator, and the projection onto
(unnormalized) maximally entangled vector are
respectively given by
\begin{equation}
\begin{aligned}
\mathbb{I}
&=
\sum_{i,j=1}^d
\ket{i}\!\bra{i}\otimes\ket{j}\!\bra{j},
\\
\mathbb{F}
&=
\sum_{i,j=1}^d
\ket{i}\!\bra{j}\otimes\ket{j}\!\bra{i},
\\
\ket{u}
&=
\sum_{i=1}^{d}
\ket{i}\otimes\ket{i},
\qquad
\ket{u}\!\bra{u}
=
\sum_{i,j=1}^{d}
\ket{i}\!\bra{j}\otimes\ket{i}\!\bra{j}.
\end{aligned}
\label{eq:basic_operators}
\end{equation}
Let

\[
\ket{v}
= \sum_{i=1}^n \ket{a_i} \otimes \ket{b_i}
\in \mathbb{C}^d \otimes \mathbb{C}^d,
\]
where the set
$\{\ket{a_1},\dots,\ket{a_n},\ket{b_1},\dots,\ket{b_n}\}$
is linearly independent in $\mathbb{C}^{d}$. Define the
operator
\begin{equation}
\gamma
= \mathbb{I} + \mathbb{F} + \varepsilon\, \ket{v}\!\bra{v},
\qquad \varepsilon > 0.
\end{equation}

As shown by Cariello~\cite{Cariello2019}, this construction
yields a family of states $\gamma$ whose Schmidt number
satisfies

\[
SN(\gamma) = n,
\qquad
1 \le n \le \Big\lceil \frac{d-1}{2} \Big\rceil,
\]
for every $\varepsilon > 0$, and whose partial transpose
$\gamma^\Gamma$ is positive
under partial transposition (PPT). Consequently, whenever $n \ge 2$ (which implies $d \ge 2n \ge 4$), choosing $\varepsilon > 0$ sufficiently small
produces bipartite states $\gamma$ that are both PPT and entangled. This provides an explicit family of PPT entangled states with arbitrarily large Schmidt number, which we will use to construct faithful probes for AAQPT.

To verify that the bound entangled state $\gamma$
constructed above provides a valid probe for AAQPT, we
analyze the action of the realignment map
$\mathcal{R}(\cdot)$ on the operators $\mathbb{I}$,
$\mathbb{F}$, and $\ket{u}\!\bra{u}$. A direct computation
shows that
\begin{equation}
\begin{aligned}
\mathcal{R}(\mathbb{I})
&= \sum_{i,j=1}^d \ket{i}\!\bra{j} \otimes \ket{i}\!\bra{j}
= \ket{u}\!\bra{u}, \\
\mathcal{R}(\mathbb{F})
&= \sum_{i,j=1}^d \ket{i}\!\bra{j} \otimes \ket{j}\!\bra{i}
= \mathbb{F}, \\
\mathcal{R}(\ket{u}\!\bra{u})
&= \sum_{i,j=1}^d \ket{i}\!\bra{i} \otimes \ket{j}\!\bra{j}
= \mathbb{I}.
\end{aligned}
\label{eq:realignment_basic}
\end{equation}
Thus, the subspace spanned by $\Big\{ \mathbb{I}, \mathbb{F},
\ket{u}\!\bra{u} \Big\}$
is invariant under the realignment operation.

We now show that the realigned operator $\mathcal{R}(\gamma)$ is
invertible. By linearity of $\mathcal{R}$ and using the
relations~\eqref{eq:realignment_basic}, we obtain
\begin{equation}
\mathcal{R}(\gamma)
= \ket{u}\!\bra{u} + \mathbb{F} + \varepsilon\,
\mathcal{R}(\ket{v}\!\bra{v}).
\end{equation}

To see that $\ket{u}\!\bra{u}+\mathbb{F}$ is invertible,
decompose $\mathbb{C}^d\otimes\mathbb{C}^d$ into its
symmetric and antisymmetric subspaces. Since
$\mathbb{F}$ acts as $\mathbb{I}$ on the symmetric
subspace and as $-\mathbb{I}$ on the antisymmetric
subspace, while $\ket{u}$ belongs to the symmetric
subspace, the operator
$\ket{u}\!\bra{u}+\mathbb{F}$ restricts to
$\mathbb{I}+\ket{u}\!\bra{u}$ on the symmetric subspace
and to $-\mathbb{I}$ on the antisymmetric subspace. The
operator $\mathbb{I}+\ket{u}\!\bra{u}$ is positive
definite, hence invertible, and $-\mathbb{I}$ is clearly
invertible. Therefore
$\ket{u}\!\bra{u}+\mathbb{F}$ is invertible.
Since the operator $\ket{u}\!\bra{u} + \mathbb{F}$ is
invertible, it follows by continuity that, for sufficiently
small $\varepsilon > 0$, the perturbed operator
$\mathcal{R}(\gamma)$
remains invertible. The explicit form of
$\mathcal{R}(\ket{v}\!\bra{v})$ is irrelevant, since its
contribution can be made arbitrarily small. This establishes
that $\gamma$ is a PPT entangled state whose realignment is
invertible, thereby completing the construction of a bound
entangled state suitable for AAQPT.

Having established the existence of bound entangled states
that can serve as faithful probes for AAQPT, we now present
an explicit example. The state considered below does not
belong to the family $\gamma$ constructed above, but
provides another instance of an interesting bound entangled state that
can be used as a faithful probe for AAQPT. Specifically, we
consider the PPT entangled state
$\rho_{\mathrm{CCNR}} \in \mathbb{C}^4 \otimes
\mathbb{C}^4$, which attains the maximal violation of the
CCNR criterion among PPT entangled states in
$4\otimes4$ dimensions~\cite{Lukacs2025} and has also been
shown to possess operational significance in
prepare-and-measure scenarios~\cite{Marton2025}. The state
$\rho_{\mathrm{CCNR}}$ is defined as
\begin{equation}
\label{useful aaqpt}
\rho_{\mathrm{CCNR}}
=
\sum_{i=1}^4 p_i \,
|\Psi_i\rangle\!\langle \Psi_i|_{AB}
\otimes
\varrho_{A'B'}^{(i)},
\end{equation}
where
\[
p_1 = p_2 = p_3 = \frac{1}{6},
\qquad
p_4 = \frac{1}{2}.
\]

The Bell states are
\[
\ket{\Phi^\pm}
=
\frac{1}{\sqrt{2}}(\ket{00} \pm \ket{11}),
\qquad
\ket{\Psi^\pm}
=
\frac{1}{\sqrt{2}}(\ket{01} \pm \ket{10}),
\]
and the states on $A'B'$ are
\begin{align}
\varrho_{A'B'}^{(1)} &= \ket{\Psi^+}\bra{\Psi^+},
\qquad
\varrho_{A'B'}^{(2)} = \ket{\Psi^-}\bra{\Psi^-}, \nonumber\\
\varrho_{A'B'}^{(3)} &= \ket{\Phi^+}\bra{\Phi^+},
\qquad
\varrho_{A'B'}^{(4)}
=
\frac{1}{3}
\left(
\mathbb{I}_4
-
\ket{\Phi^-}\bra{\Phi^-}
\right).
\end{align}


We employ the CCNR criterion to establish that the
state is bound entangled.
This state after realignment has the following singular
values (SV) :
\begin{align*}
    \mbox{SV}\left[\mathcal{R}(\rho_{CCNR})\right] &= \left\{
 \frac{1}{12}, \frac{1}{12}, \frac{1}{12},
\frac{1}{12}, \frac{1}{12}, \frac{1}{12}, \frac{1}{12},
\right. \\
    &\quad \left.
    \frac{1}{12},
 \frac{1}{12}, \frac{1}{12}, \frac{1}{12},
\frac{1}{12}, \frac{1}{12}, \frac{1}{12}, \frac{1}{12},
\frac{1}{4}
    \right\}
\end{align*}
Since all singular values of $\mathcal{R}(\rho_{\mathrm{CCNR}})$ are strictly positive, the realigned matrix is invertible. Consequently, $\rho_{\mathrm{CCNR}}$ is faithful and can therefore serve as a probe state for AAQPT. Furthermore,
\begin{equation}
\left\| \mathcal{R}(\rho_{\mathrm{CCNR}}) \right\|_{\mathrm{tr}}
=1.5,
\end{equation}
which exceeds unity and hence certifies entanglement of $\rho_{\mathrm{CCNR}}$.

Having established that PPT entangled states can be faithful, we now ask whether states with stronger violations of the CCNR criterion can also be useful for AAQPT. To investigate this, we maximize $\|\mathcal{R}(\rho)\|_{\mathrm{tr}}$ over PPT states and then check whether the resulting states are faithful.

Interestingly, the examples considered here suggest that PPT states with maximal CCNR violation generally have maximal OSR and are therefore faithful for AAQPT. However, it remains unclear whether maximal CCNR violation over PPT states necessarily implies maximal operator Schmidt rank. We therefore verify faithfulness separately for the states obtained below and compare their conditioning in Sec.~\ref{comparison_aaqpt}.

We numerically maximize the CCNR norm over PPT states
in $d\otimes d$ dimensions by solving
\begin{equation}
\begin{aligned}
\max_{\rho}\quad&
\|\mathcal{R}(\rho)\|_{\mathrm{tr}}\\
\text{subject to}\quad&
\rho\geq0,\quad
\rho^\Gamma\geq0,\quad
\Tr(\rho)=1,
\end{aligned}
\end{equation}
where $\rho^\Gamma$ denotes the partial transpose
of $\rho$.

Following~\cite{Lukacs2025}, we use the dual
representation of the trace norm,
\begin{equation}
\|X\|_{\mathrm{tr}}
=
\max_{Y:\,YY^\dagger\leq\mathbb{I}}
\mathrm{Tr}(X^\dagger Y).
\end{equation}
This allows us to perform
an iterative see-saw optimization, alternating
between the state $\rho$ and the auxiliary
matrix $Y$. For fixed $Y$, the optimization
over $\rho$ is a semidefinite program subject
to the PPT constraints. For fixed $\rho$,
the optimal $Y$ is obtained from the singular
value decomposition of $\mathcal{R}(\rho)$.
The procedure is repeated until the CCNR
norm converges.

For the $3\otimes3$ case, this gives a PPT entangled state with maximal operator Schmidt rank and hence faithful for AAQPT. Its CCNR violation agrees with the maximum found numerically in~\cite{Lukacs2025}. The density matrix of this state is given in Appendix~\ref{appendix_3x3}. More generally, the same approach can be extended to arbitrary $d\otimes d$ dimensions to search for PPT entangled states with maximal operator Schmidt rank and strong CCNR violation. However, finding such states becomes increasingly difficult as the local dimension increases. An alternative is to restrict the search to Bloch-diagonal PPT states, as considered by Tavakoli \textit{et al.}~\cite{Tavakoli2026}, which provides a simpler class of states to investigate. In their paper, the $4\otimes4$ Bloch-diagonal PPT entangled state also has maximal operator Schmidt rank.

\section{Effect of local operations on
AAQPT}\label{sec:local_filtering_AAQPT}

Since the proposal by Gisin~\cite{Gisin1996}, local filtering has been known to reveal operational advantages in several quantum information tasks. For instance, appropriate filters can be used for entanglement purification~\cite{Horodecki1997e,Kwiat2001,Wang2006}, reveal hidden nonlocality~\cite{Masanes2008,Jones2020}, study EPR steering~\cite{Ku2022}, and uncover hidden steerability~\cite{Pramanik2019}. It has also been used to activate hidden teleportation power (HTP)~\cite{Li2021}. 
Here, we investigate the effects of local filtering operations on AAQPT.
Consider the local filtering operation:
\begin{equation}
\rho'=\frac{(A\otimes B)\,\rho\, (A\otimes B)^{\dagger}}
{\operatorname{Tr}[(A\otimes B)\,\rho\, (A\otimes
B)^{\dagger}]},
\label{LQCC}
\end{equation}
where the local filters satisfy $A^{\dagger}A\leq \mathbb{I}$ and
$B^{\dagger}B\leq \mathbb{I}$.

The question is a natural one, because every task listed above is improved by a
filter that raises a quantity linear in the correlations of the state, and the
realignment norm responds to filtering in the same way. Process reconstruction,
in contrast, proceeds by inverting the realigned probe matrix, so a probe is
improved only when the smallest singular value of $\mathcal{R}(\rho)$ is
improved. A filter raises the realignment norm by concentrating weight into
fewer singular directions, which is the one operation that an inversion cannot
tolerate. The two objectives are therefore in opposition, and the remainder of
this section makes that opposition quantitative.

In contrast to unitary local operations, which leave the
trace norm of the
realigned operator invariant, general local operations can
significantly
change the value of the quantity $\|\mathcal{R}(\rho)\|_{\mathrm{tr}}$.

Quantum operations implementable by local operations and
classical communication are generated by the following
elementary
operations~\cite{Donald2002}: (i) attaching uncorrelated
ancillas,
(ii) tracing out subsystems, (iii) performing local
unitaries, and
(iv) carrying out Lüders--von Neumann projective
measurements.
These transformations affect the realignment norm in
distinct ways.

\begin{proposition}[Rudolph~\cite{Rudolph2003}]
\label{p9_rewrite}
The quantity $\|\mathcal{R}(\rho)\|_{\mathrm{tr}}$ satisfies the
following properties:
(i) it is invariant under local unitary transformations;
(ii) it is non-increasing under the addition of uncorrelated
ancillas;
(iii) it is non-increasing under Lüders--von Neumann
projective measurements;
(iv) it may increase, decrease, or remain unchanged under
tracing out
subsystems.
\end{proposition}

As an immediate consequence, the CCNR criterion is not
invariant under general local operations. In particular,
there exist states satisfying
$\|\mathcal{R}(\rho)\|_{\mathrm{tr}}\leq 1$
whose CCNR norm exceeds unity after local filtering
and postselection.

\medskip

For AAQPT, however, the relevant property is the
faithfulness of the probe, which requires its realigned
matrix to have full rank. We therefore examine how
local filtering affects the operator Schmidt rank
of a bipartite state.

Consider local filters $A$ and $B$ acting on the two
subsystems. The normalized state after successful
filtering is
\begin{equation}
\rho'=
\frac{(A\otimes B)\rho_{AB}(A\otimes B)^\dagger}
{\Tr[(A\otimes B)\rho_{AB}(A\otimes B)^\dagger]}.
\label{eq:filtered_state}
\end{equation}
Using the definition of realignment, one obtains
\begin{equation}
\mathcal{R}(\rho')
=
\frac{
(A\otimes\bar A)\,
\mathcal{R}(\rho_{AB})\,
(B\otimes\bar B)^T
}{
\Tr[(A\otimes B)\rho_{AB}(A\otimes B)^\dagger]
},
\label{eq:filter-law}
\end{equation}
where the bar denotes complex conjugation in the
computational basis.

Equation~(\ref{eq:filter-law}) shows that the effect
of local filtering on the realigned matrix is simply
multiplication by $A\otimes\bar A$ and
$(B\otimes\bar B)^T$ from the left and right,
respectively. This immediately gives the following
result.

\begin{proposition}\label{prop2}
A faithful bipartite state becomes unfaithful under
local filtering if either of the filters is rank
deficient, provided the filtering succeeds.
\end{proposition}

\begin{proof}
From Eq.~(\ref{eq:filter-law}), the rank of the
realigned matrix satisfies
\begin{equation}
\operatorname{rank}\mathcal{R}(\rho')
\leq
\min\left\{
\operatorname{rank}(A)^2,
\operatorname{rank}(B)^2
\right\}.
\end{equation}
If either $A$ or $B$ is rank deficient, then
\begin{equation}
\operatorname{OSR}(\rho')
=
\operatorname{rank}\mathcal{R}(\rho')
<d^2.
\end{equation}
Thus, the filtered state no longer has maximal
operator Schmidt rank and cannot be faithful
for AAQPT.
\end{proof}

An example of this behaviour can be found in the
local filtering of Werner states~\cite{Werner1989}.
Recall that
\begin{equation}
\rho_W
=
\frac{2v}{d(d+1)}P_{+}
+
\frac{2(1-v)}{d(d-1)}P_{-},
\qquad v\in[0,1],
\end{equation}
where $P_{\pm}=(\mathbb{I}_{d^2}\pm\mathbb{F})/2$
are the projectors onto the symmetric and
antisymmetric subspaces, respectively, and
$\mathbb{F}=\sum_{i,j=1}^{d}
\ket{i}\bra{j}\otimes\ket{j}\bra{i}$
is the swap operator.
For $d>2$, Werner states satisfy the reduction
criterion~\cite{Horodecki1999c} and are not useful
for standard teleportation. However, local filtering
can activate hidden teleportation power
(HTP)~\cite{Li2021}.

In particular, the local filters considered
in~\cite{Li2021} are
\begin{equation}
A_W=\sigma_z\oplus\mathbf{0}_{d-2},
\qquad
B_W=\sigma_x\oplus\mathbf{0}_{d-2},
\end{equation}
where $\sigma_x$ and $\sigma_z$ are Pauli operators.
These filters act on a two-dimensional subspace
of each subsystem and remove the remaining
components. The resulting state is
\begin{align}
\nonumber
\rho_{Wf}
&=
\frac{1}{N}\Big[
(d+1)(1-v)
\ket{\Phi^+_2}\bra{\Phi^+_2}
\\
&\qquad
+v(d-1)\big(
\mathbb{I}_4-
\ket{\Phi^+_2}\bra{\Phi^+_2}
\big)\Big]
\oplus\mathbf{0}_{d^2-4},
\label{eq:rhoWf_block}
\end{align}
where
$\ket{\Phi^+_2}=(\ket{00}+\ket{11})/\sqrt{2}$
and $N=(d+1)(1-v)+3v(d-1)$.

Such filtering can increase the CCNR norm of
the state. For instance, at $v=0$,
\begin{equation}
\|\mathcal{R}(\rho_W)\|_{\mathrm{tr}}
=1+\frac{2}{d},
\qquad
\|\mathcal{R}(\rho_{Wf})\|_{\mathrm{tr}}=2.
\end{equation}
Nevertheless, the filtered state is no longer
faithful for AAQPT when $d>2$. Since both
local filters have rank two, Proposition~\ref{prop2}
gives
\begin{equation}
\operatorname{OSR}(\rho_{Wf})
=
\operatorname{rank}\mathcal{R}(\rho_{Wf})
\leq 4<d^2.
\end{equation}
Consequently,
$\sigma_{\min}(\mathcal{R}(\rho_{Wf}))=0$,
and the filtered state cannot be used for
complete channel reconstruction. Thus, although
local filtering can improve the teleportation
properties and increase the CCNR norm of a
state, the same filtering can make it unsuitable
for AAQPT.

The situation is different when both local filters are
invertible. In this case, Eq.~\eqref{eq:filter-law}
shows that the operator Schmidt rank remains unchanged,
so a faithful probe remains faithful after filtering.
However, faithfulness alone does not tell us how well
conditioned the reconstruction will be. To examine this,
we consider the condition number
$\kappa(X)=\sigma_{\max}(X)/\sigma_{\min}(X)$,
where $\sigma_{\max}$ and $\sigma_{\min}$ denote
the largest and smallest singular values of $X$,
respectively.

\begin{proposition}
\label{prop:invertible}
If $A$ and $B$ are invertible, then the filtered
state $\rho'$ is faithful whenever $\rho_{AB}$
is faithful. Furthermore,
\begin{equation}
\begin{aligned}
\frac{\kappa(\mathcal{R}(\rho_{AB}))}
{\kappa(A)^2\kappa(B)^2}
&\leq
\kappa(\mathcal{R}(\rho'))
\\
&\leq
\kappa(A)^2\kappa(B)^2
\kappa(\mathcal{R}(\rho_{AB})).
\end{aligned}
\end{equation}
\end{proposition}

To see this, note that the matrices
$A\otimes\bar A$ and $(B\otimes\bar B)^T$
in Eq.~\eqref{eq:filter-law} are invertible
whenever $A$ and $B$ are invertible.
Their multiplication with
$\mathcal{R}(\rho_{AB})$ therefore preserves
its rank. The bounds on the condition number
follow from the corresponding inequalities
for products of matrices, together with
$\kappa(A\otimes\bar A)=\kappa(A)^2$
and $\kappa(B\otimes\bar B)=\kappa(B)^2$.
The normalization of the filtered state does
not affect these bounds, since the condition
number is invariant under multiplication
by a nonzero scalar.

For local unitary operations, the condition
numbers of both filters are one, and the
conditioning of the realigned matrix remains
unchanged. General invertible filters need
not have this property. Depending on the
choice of filters, the smallest singular
value and the condition number of the
realigned matrix may change, even though
the state remains faithful. Thus, preserving
faithfulness does not necessarily mean
preserving the stability of channel
reconstruction.

There is also a restriction on what can be
achieved by local filtering, irrespective
of whether the filters are invertible.
From Eq.~\eqref{eq:filter-law}, the realigned
matrix of the filtered state is obtained
by multiplying $\mathcal{R}(\rho_{AB})$
from the left and right. Since such
multiplication cannot increase matrix rank,
we have the following proposition.

\begin{proposition}
\label{prop:monotone}
For any successful local filtering operation,
\begin{equation}
\operatorname{OSR}(\rho')
\leq
\operatorname{OSR}(\rho_{AB}).
\end{equation}
In particular, an unfaithful state cannot be
made faithful by local filtering.
\end{proposition}

The proposition holds for both invertible
and rank-deficient filters. Invertible filters
preserve the operator Schmidt rank, whereas
rank-deficient filters can reduce it.
Consequently, if the realigned matrix of
the original state is not full rank, no
choice of local filters can make it
invertible. This is particularly relevant
when local filtering is used to increase
the CCNR norm, since an increase in the
trace norm does not imply an increase in
the operator Schmidt rank or an improvement
in the conditioning of the reconstruction.

For AAQPT, the distinction is therefore
between improving the correlations of a
state as measured by the CCNR norm and
improving its usefulness for channel
reconstruction. The Werner-state example
discussed above illustrates how a filter
can increase the CCNR value while destroying
faithfulness altogether. Invertible filters
avoid this loss of faithfulness, but their
effect on the conditioning must still be
examined separately.


\section{Separable versus entangled probes for AAQPT}
\label{comparison_aaqpt}

It is known that separable states can also be used
for AAQPT~\cite{D_Ariano2003}. It is therefore
interesting to ask whether bound entangled states
offer any advantage over separable probes. In
particular, we compare their performance in terms
of the stability of channel reconstruction.

For a faithful probe, reconstruction involves
inverting its realigned matrix
$\mathcal{R}(\rho)$~\cite{D_Ariano2001,Zhang2023}.
The inversion can amplify experimental errors,
especially when the smallest singular value is
close to zero. We therefore consider the smallest
singular value and the condition number,
\begin{equation}
\label{eq:figures}
\sigma_{\min}(\mathcal{R}(\rho)),
\qquad
\kappa(\mathcal{R}(\rho))
=
\frac{\sigma_{\max}(\mathcal{R}(\rho))}
{\sigma_{\min}(\mathcal{R}(\rho))}.
\end{equation}
A larger $\sigma_{\min}$ and a smaller $\kappa$
correspond to better-conditioned reconstruction.
These quantities need not follow the CCNR norm,
which depends on the sum of the singular values
rather than their smallest value.

We first consider the family of isotropic states,
\begin{equation}
\label{isotropic-state-defn}
\rho_{\mathrm{iso}}
=
\frac{1-\alpha}{d^2}\mathbb{I}
+
\alpha\ket{\Phi^+}\bra{\Phi^+},
\qquad
-\frac{1}{d^2-1}\leq\alpha\leq1.
\end{equation}
These states are separable for
$\alpha\leq1/(d+1)$ and entangled otherwise.
Using Eq.~\eqref{eq:realignment_basic}, we obtain
\[
\mathcal{R}(\rho_{\mathrm{iso}})
=
\frac{1-\alpha}{d^2}\ket{u}\bra{u}
+
\frac{\alpha}{d}\mathbb{I}.
\]
The singular values are $1/d$ once and
$|\alpha|/d$ with multiplicity $d^2-1$.
For the maximally entangled state at $\alpha=1$,
all singular values are equal to $1/d$, giving
$\sigma_{\min}=1/d$ and $\kappa=1$.
At the separability boundary,
$\alpha=1/(d+1)$, we instead obtain
$\sigma_{\min}=1/[d(d+1)]$ and $\kappa=d+1$.
We will show below that these are the best values
attainable by any separable probe.

For comparison, consider the antisymmetric
Werner state discussed in the previous section,
$\rho_W=(\mathbb{I}-\mathbb{F})/(d^2-d)$.
Its realigned matrix has singular values
$1/d$ once and $1/(d^2-d)$ with multiplicity
$d^2-1$, giving
$\sigma_{\min}=1/(d^2-d)$ and $\kappa=d-1$.
For $d=4$, these are precisely the values
obtained for the bound entangled state
$\rho_{\mathrm{CCNR}}$ considered in Sec.~III.
In fact, both states have the same realigned
singular spectrum, despite one being distillable
and the other bound entangled.

We now show that there is a general bound on
the conditioning of separable probes.
For any density operator $\rho$, let
\begin{equation}
\ket{\omega}
=
\frac{1}{\sqrt{d}}
\sum_{i=1}^{d}\ket{i}\otimes\ket{i}
\end{equation}
be the normalized maximally entangled vector.
From the definition of realignment, we have
\begin{equation}
\label{eq:omega}
\bra{\omega}\mathcal{R}(\rho)\ket{\omega}
=
\frac{\Tr(\rho)}{d}
=
\frac{1}{d}.
\end{equation}
Since the largest singular value of a matrix
is at least as large as the absolute value of
any expectation value taken over a normalized
vector, it follows that
\begin{equation}
\sigma_{\max}(\mathcal{R}(\rho))
\geq \frac{1}{d}.
\end{equation}
Combining this inequality with the CCNR
criterion gives the following result.

\begin{proposition}
\label{prop:ceiling}
For every separable state $\rho_{\mathrm{sep}}$
on $\mathbb{C}^d\otimes\mathbb{C}^d$,
\begin{equation}
\label{eq:ceiling}
\sigma_{\min}(\mathcal{R}(\rho_{\mathrm{sep}}))
\leq\frac{1}{d(d+1)},
\qquad
\kappa(\mathcal{R}(\rho_{\mathrm{sep}}))
\geq d+1.
\end{equation}
Any state violating either inequality must
therefore be entangled.
\end{proposition}

\begin{proof}
Let $s_1\geq\cdots\geq s_{d^2}$ be the
singular values of
$\mathcal{R}(\rho_{\mathrm{sep}})$.
Since $s_1\geq1/d$ and the CCNR criterion
gives $\sum_l s_l\leq1$, we have
\[
\begin{aligned}
\sigma_{\min}
&\leq\frac{1-s_1}{d^2-1}
&\leq\frac{1-1/d}{d^2-1}
=\frac{1}{d(d+1)}.
\end{aligned}
\]
Together with $\sigma_{\max}\geq1/d$,
this gives $\kappa\geq d+1$.
For rank-deficient realigned matrices,
the condition number is taken to be infinite.
\end{proof}

Both inequalities are saturated by the isotropic
state at $\alpha=1/(d+1)$. The bounds are
therefore tight.

For $d=4$, every separable probe satisfies
$\sigma_{\min}\leq1/20$ and $\kappa\geq5$.
In contrast, the bound entangled state
$\rho_{\mathrm{CCNR}}$ has
$\sigma_{\min}=1/12$ and $\kappa=3$.
Thus, the bound entangled state is better
conditioned than any separable probe in
$\mathbb{C}^4\otimes\mathbb{C}^4$, under
both measures.

Table~\ref{tab:faithfulness} summarizes the
comparison. While a maximally entangled probe
remains optimal, the bound entangled state
outperforms every separable probe and has the
same conditioning as the antisymmetric Werner
state. This shows that distillability does not
determine the conditioning of a state for AAQPT.

\begin{table}[t]
\centering
\renewcommand{\arraystretch}{1.2}
\begin{tabular}{lccc}
\hline\hline
State & Class & $\sigma_{\min}$ & $\kappa$ \\
\hline
$\ket{\Phi^+}\!\bra{\Phi^+}$
    & Maximally entangled & $1/4$ & $1$ \\
Antisymmetric Werner
    & Distillable entangled & $1/12$ & $3$ \\
$\rho_{\mathrm{CCNR}}$
    & Bound entangled & $1/12$ & $3$ \\
Isotropic ($\alpha=1/5$)
    & Separable & $1/20$ & $5$ \\
\hline
Separable bound
    & --- & $\leq 1/20$ & $\geq 5$ \\
\hline\hline
\end{tabular}
\caption{Conditioning of representative AAQPT probes in
$\mathbb{C}^4\otimes\mathbb{C}^4$.
The separable bounds from Proposition~\ref{prop:ceiling}
are tight and are attained by the isotropic state
at $\alpha=1/5$.}
\label{tab:faithfulness}
\end{table}

\section{Conclusion}\label{conclusion}

In this work, we have shown that bound entangled states
can be used for ancilla-assisted quantum process tomography.
We constructed a family of faithful PPT entangled states
with arbitrarily large Schmidt number and examined
explicit examples in $3\otimes3$ and $4\otimes4$ dimensions.
Furthermore, we showed that the $4\otimes4$ bound entangled
state provides better-conditioned channel reconstruction
than any separable probe. In particular, we derived
the bounds $\sigma_{\min}\leq 1/[d(d+1)]$ and
$\kappa\geq d+1$ for all separable states, both of
which are surpassed by our bound entangled example.
Thus, although entanglement is not necessary for AAQPT,
bound entanglement can offer an advantage over separable
states.

We also investigated the effect of local filtering on
the faithfulness of a state. While local filtering can
increase the CCNR norm and improve certain operational
properties, it need not improve the performance of a
probe for AAQPT. In particular, rank-deficient filters
destroy faithfulness, whereas invertible filters preserve
it, although they may change the conditioning of the
reconstruction. We further showed that no local filter
can increase the operator Schmidt rank of a state,
and hence an unfaithful state cannot be made faithful
by local filtering.

We hope that the present work will also encourage
further investigation into the structure of states
with maximal operator Schmidt rank, particularly
among bound entangled states. One interesting
observation is that the PPT entangled states with
maximal CCNR violation considered here also have
maximal operator Schmidt rank. This raises the
question of whether there is a more general connection
between these two properties. Although maximal
CCNR violation does not by itself guarantee
faithfulness, the examples suggest that the relation
between the two deserves further investigation.

It would be interesting to determine whether PPT entangled
states with maximal CCNR violation necessarily have maximal
operator Schmidt rank, or whether counterexamples exist.
A better understanding of this connection may also help
in characterizing bound entangled states that are faithful
for AAQPT, particularly in higher dimensions. We also
hope that our results will motivate experimental studies
of quantum process tomography using bound entangled probes.

\section*{Acknowledgment}

I would like to thank Arvind and Daniel Cariello
for helpful discussions and comments. I also thank
the anonymous reviewers for their useful comments
and suggestions.

\bibliographystyle{apsrev4-2}
\bibliography{ref}

@article{Werner1989,
  title = {Quantum states with Einstein-Podolsky-Rosen correlations admitting a hidden-variable model},
  author = {Werner, Reinhard F.},
  journal = {Phys. Rev. A},
  volume = {40},
  issue = {8},
  pages = {4277--4281},
  numpages = {0},
  year = {1989},
  month = {Oct},
  publisher = {American Physical Society},
  doi = {10.1103/PhysRevA.40.4277},
  url = {https://link.aps.org/doi/10.1103/PhysRevA.40.4277}
}

@article{Gisin1996,
title = {Hidden quantum nonlocality revealed by local filters},
journal = {Physics Letters A},
volume = {210},
number = {3},
pages = {151-156},
year = {1996},
issn = {0375-9601},
doi = {https://doi.org/10.1016/S0375-9601(96)80001-6},
url = {https://www.sciencedirect.com/science/article/pii/S0375960196800016},
author = {N. Gisin}
}

@article{Horodecki1999c,
  title = {Reduction criterion of separability and limits for a class of distillation protocols},
  author = {Horodecki, Micha\l{} and Horodecki, Pawe\l{}},
  journal = {Phys. Rev. A},
  volume = {59},
  issue = {6},
  pages = {4206--4216},
  numpages = {0},
  year = {1999},
  month = {Jun},
  publisher = {American Physical Society},
  doi = {10.1103/PhysRevA.59.4206},
  url = {https://link.aps.org/doi/10.1103/PhysRevA.59.4206}
}

@article{Lu2022,
   title={Not All Entangled States are Useful for Ancilla‐Assisted Quantum Process Tomography},
   volume={534},
   ISSN={1521-3889},
   url={http://dx.doi.org/10.1002/andp.202100550},
   DOI={10.1002/andp.202100550},
   number={5},
   journal={Annalen der Physik},
   publisher={Wiley},
   author={Lu, Guo‐Dong and Zhang, Zhou and Dai, Yue and Dong, Yu‐Li and Zhang, Cheng‐Jie},
   year={2022},
   month=feb }

@article{Chuang1997,
   title={Prescription for experimental determination of the dynamics of a quantum black box},
   volume={44},
   ISSN={1362-3044},
   url={http://dx.doi.org/10.1080/09500349708231894},
   DOI={10.1080/09500349708231894},
   number={11–12},
   journal={Journal of Modern Optics},
   publisher={Informa UK Limited},
   author={Chuang, Isaac L. and Nielsen, M. A.},
   year={1997},
   month=nov, pages={2455–2467} }

@article{Poyatos1997,
  title = {Complete Characterization of a Quantum Process: The Two-Bit Quantum Gate},
  author = {Poyatos, J. F. and Cirac, J. I. and Zoller, P.},
  journal = {Phys. Rev. Lett.},
  volume = {78},
  issue = {2},
  pages = {390--393},
  numpages = {0},
  year = {1997},
  month = {Jan},
  publisher = {American Physical Society},
  doi = {10.1103/PhysRevLett.78.390},
  url = {https://link.aps.org/doi/10.1103/PhysRevLett.78.390}
}

@article{D_Ariano2001,
  title = {Quantum Tomography for Measuring Experimentally the Matrix Elements of an Arbitrary Quantum Operation},
  author = {D'Ariano, G. M. and Lo Presti, P.},
  journal = {Phys. Rev. Lett.},
  volume = {86},
  issue = {19},
  pages = {4195--4198},
  numpages = {0},
  year = {2001},
  month = {May},
  publisher = {American Physical Society},
  doi = {10.1103/PhysRevLett.86.4195},
  url = {https://link.aps.org/doi/10.1103/PhysRevLett.86.4195}
}

@article{D_Ariano2003,
   title={Imprinting Complete Information about a Quantum Channel on its Output State},
   volume={91},
   ISSN={1079-7114},
   url={http://dx.doi.org/10.1103/PhysRevLett.91.047902},
   DOI={10.1103/physrevlett.91.047902},
   number={4},
   journal={Physical Review Letters},
   publisher={American Physical Society (APS)},
   author={D’Ariano, Giacomo Mauro and Lo Presti, Paoloplacido},
   year={2003},
   month=jul }

@article{Altepeter2003,
   title={Ancilla-Assisted Quantum Process Tomography},
   volume={90},
   ISSN={1079-7114},
   url={http://dx.doi.org/10.1103/PhysRevLett.90.193601},
   DOI={10.1103/physrevlett.90.193601},
   number={19},
   journal={Physical Review Letters},
   publisher={American Physical Society (APS)},
   author={Altepeter, J. B. and Branning, D. and Jeffrey, E. and Wei, T. C. and Kwiat, P. G. and Thew, R. T. and O’Brien, J. L. and Nielsen, M. A. and White, A. G.},
   year={2003},
   month=may }

@article{Mohseni2006,
  title = {Direct Characterization of Quantum Dynamics},
  author = {Mohseni, M. and Lidar, D. A.},
  journal = {Phys. Rev. Lett.},
  volume = {97},
  issue = {17},
  pages = {170501},
  numpages = {4},
  year = {2006},
  month = {Oct},
  publisher = {American Physical Society},
  doi = {10.1103/PhysRevLett.97.170501},
  url = {https://link.aps.org/doi/10.1103/PhysRevLett.97.170501}
}

@article{Mohseni2008,
  title = {Quantum-process tomography: Resource analysis of different strategies},
  author = {Mohseni, M. and Rezakhani, A. T. and Lidar, D. A.},
  journal = {Phys. Rev. A},
  volume = {77},
  issue = {3},
  pages = {032322},
  numpages = {15},
  year = {2008},
  month = {Mar},
  publisher = {American Physical Society},
  doi = {10.1103/PhysRevA.77.032322},
  url = {https://link.aps.org/doi/10.1103/PhysRevA.77.032322}
}

@book{Nielsen2010,
  title={Quantum computation and quantum information},
  author={Nielsen, Michael A and Chuang, Isaac L},
  year={2010},
  publisher={Cambridge university press}
}

@article{Horodecki1998,
  title = {Mixed-State Entanglement and Distillation: Is there a ``Bound'' Entanglement in Nature?},
  author = {Horodecki, Micha\l{} and Horodecki, Pawe\l{} and Horodecki, Ryszard},
  journal = {Phys. Rev. Lett.},
  volume = {80},
  issue = {24},
  pages = {5239--5242},
  numpages = {0},
  year = {1998},
  month = {Jun},
  publisher = {American Physical Society},
  doi = {10.1103/PhysRevLett.80.5239},
  url = {https://link.aps.org/doi/10.1103/PhysRevLett.80.5239}
}

@article{Chi2007,
  title = {Bound entangled states with a nonzero distillable key rate},
  author = {Chi, Dong Pyo and Choi, Jeong Woon and Kim, Jeong San and Kim, Taewan and Lee, Soojoon},
  journal = {Phys. Rev. A},
  volume = {75},
  issue = {3},
  pages = {032306},
  numpages = {7},
  year = {2007},
  month = {Mar},
  publisher = {American Physical Society},
  doi = {10.1103/PhysRevA.75.032306},
  url = {https://link.aps.org/doi/10.1103/PhysRevA.75.032306}
}

@article{Mishra2020,
  title = {Increasing distillable key rate from bound entangled states by using local filtration},
  author = {Mishra, Mayank and Sengupta, Ritabrata and Arvind},
  journal = {Phys. Rev. A},
  volume = {102},
  issue = {3},
  pages = {032415},
  numpages = {10},
  year = {2020},
  month = {Sep},
  publisher = {American Physical Society},
  doi = {10.1103/PhysRevA.102.032415},
  url = {https://link.aps.org/doi/10.1103/PhysRevA.102.032415}
}

@article{Toth2018,
  title = {Quantum States with a Positive Partial Transpose are Useful for Metrology},
  author = {T\'oth, G\'eza and V\'ertesi, Tam\'as},
  journal = {Phys. Rev. Lett.},
  volume = {120},
  issue = {2},
  pages = {020506},
  numpages = {6},
  year = {2018},
  month = {Jan},
  publisher = {American Physical Society},
  doi = {10.1103/PhysRevLett.120.020506},
  url = {https://link.aps.org/doi/10.1103/PhysRevLett.120.020506}
}

@article{Shor2003,
  title = {Superactivation of Bound Entanglement},
  author = {Shor, Peter W. and Smolin, John A. and Thapliyal, Ashish V.},
  journal = {Phys. Rev. Lett.},
  volume = {90},
  issue = {10},
  pages = {107901},
  numpages = {4},
  year = {2003},
  month = {Mar},
  publisher = {American Physical Society},
  doi = {10.1103/PhysRevLett.90.107901},
  url = {https://link.aps.org/doi/10.1103/PhysRevLett.90.107901}
}

@article{Pal2021,
  title = {Bound entangled singlet-like states for quantum metrology},
  author = {P\'al, K\'aroly F. and T\'oth, G\'eza and Bene, Erika and V\'ertesi, Tam\'as},
  journal = {Phys. Rev. Res.},
  volume = {3},
  issue = {2},
  pages = {023101},
  numpages = {18},
  year = {2021},
  month = {May},
  publisher = {American Physical Society},
  doi = {10.1103/PhysRevResearch.3.023101},
  url = {https://link.aps.org/doi/10.1103/PhysRevResearch.3.023101}
}

@article{Cariello2019,
   title={Inequalities for the Schmidt number of bipartite states},
   volume={110},
   ISSN={1573-0530},
   url={http://dx.doi.org/10.1007/s11005-019-01244-1},
   DOI={10.1007/s11005-019-01244-1},
   number={4},
   journal={Letters in Mathematical Physics},
   publisher={Springer Science and Business Media LLC},
   author={Cariello, Daniel},
   year={2019},
   month=nov, pages={827–833} }

@article{Chen2003,
author = {Chen, Kai and Wu, Ling-An},
title = {A matrix realignment method for recognizing entanglement},
year = {2003},
issue_date = {May 2003},
publisher = {Rinton Press, Incorporated},
address = {Paramus, NJ},
volume = {3},
number = {3},
issn = {1533-7146},
journal = {Quantum Info. Comput.},
month = may,
pages = {193–202},
numpages = {10}
}

@article{Donald2002,
   title={The uniqueness theorem for entanglement measures},
   volume={43},
   ISSN={1089-7658},
   url={http://dx.doi.org/10.1063/1.1495917},
   DOI={10.1063/1.1495917},
   number={9},
   journal={Journal of Mathematical Physics},
   publisher={AIP Publishing},
   author={Donald, Matthew J. and Horodecki, Michał and Rudolph, Oliver},
   year={2002},
   month=sep, pages={4252–4272} }

@article{Rudolph2003,
   title={Some properties of the computable cross-norm criterion for separability},
   volume={67},
   ISSN={1094-1622},
   url={http://dx.doi.org/10.1103/PhysRevA.67.032312},
   DOI={10.1103/physreva.67.032312},
   number={3},
   journal={Physical Review A},
   publisher={American Physical Society (APS)},
   author={Rudolph, Oliver},
   year={2003},
   month=mar }

@article{Rudolph2005,
   title={Further Results on the Cross Norm Criterion for Separability},
   volume={4},
   ISSN={1573-1332},
   url={http://dx.doi.org/10.1007/s11128-005-5664-1},
   DOI={10.1007/s11128-005-5664-1},
   number={3},
   journal={Quantum Information Processing},
   publisher={Springer Science and Business Media LLC},
   author={Rudolph, Oliver},
   year={2005},
   month=aug, pages={219–239} }

@book{Kraus1983,
  title={States, Effects, and Operations Fundamental Notions of Quantum Theory: Lectures in Mathematical Physics at the University of Texas at Austin},
  author={Kraus, Karl and B{\"o}hm, Arno and Dollard, John D and Wootters, WH},
  year={1983},
  publisher={Springer}
}

@article{Jamiolkowski1972,
  title={Linear transformations which preserve trace and positive semidefiniteness of operators},
  author={Jamio{\l}kowski, Andrzej},
  journal={Reports on mathematical physics},
  volume={3},
  number={4},
  pages={275--278},
  year={1972},
  publisher={Elsevier}
}

@article{Choi1975,
  title={Completely positive linear maps on complex matrices},
  author={Choi, Man-Duen},
  journal={Linear algebra and its applications},
  volume={10},
  number={3},
  pages={285--290},
  year={1975},
  publisher={Elsevier}
}

@article{Arrighi2004,
   title={On quantum operations as quantum states},
   volume={311},
   ISSN={0003-4916},
   url={http://dx.doi.org/10.1016/j.aop.2003.11.005},
   DOI={10.1016/j.aop.2003.11.005},
   number={1},
   journal={Annals of Physics},
   publisher={Elsevier BV},
   author={Arrighi, Pablo and Patricot, Christophe},
   year={2004},
   month=may, pages={26–52} }

@article{Lukacs2025,
   title={Iterative optimization in quantum metrology and entanglement theory using semidefinite programming},
   volume={11},
   ISSN={2058-9565},
   url={http://dx.doi.org/10.1088/2058-9565/ae24a6},
   DOI={10.1088/2058-9565/ae24a6},
   number={1},
   journal={Quantum Science and Technology},
   publisher={IOP Publishing},
   author={Luk{\'a}cs, {\'A}rp{\'a}d and Tr{\'e}nyi, R{\'o}bert and V{\'e}rtesi, Tam{\'a}s and T{\'o}th, G{\'e}za},
   year={2026},
   month=jan, pages={015042} }

@article{MullerHermesNechita2018,
  author  = {Alexander M{\"u}ller-Hermes
             and Ion Nechita},
  title   = {Operator Schmidt ranks of bipartite
             unitary matrices},
  journal = {Linear Algebra and its Applications},
  volume  = {557},
  pages   = {174--187},
  year    = {2018},
  doi     = {10.1016/j.laa.2018.07.018}
}

@article{Tavakoli2026,
   title={Unlimited quantum correlation advantage from bound entanglement},
   volume={28},
   ISSN={1367-2630},
   url={http://dx.doi.org/10.1088/1367-2630/ae45c5},
   DOI={10.1088/1367-2630/ae45c5},
   number={3},
   journal={New Journal of Physics},
   publisher={IOP Publishing},
   author={Tavakoli, Armin and Roch i Carceller, Carles and Tendick, Lucas and Vértesi, Tamás},
   year={2026},
   month=Mar, pages={034503} }

@article{Bao2024,
   title={Quantum properties of bipartite separable mixed state for ancilla-assisted process tomography},
   volume={25},
   ISSN={1573-1332},
   url={http://dx.doi.org/10.1007/s11128-025-05036-6},
   DOI={10.1007/s11128-025-05036-6},
   number={1},
   journal={Quantum Information Processing},
   publisher={Springer Science and Business Media LLC},
   author={Bao, Zhuoran and James, Daniel F. V.},
   year={2026},
   month=Jan }

@article{Roch2025,
  title = {Bound Entangled States Are Useful in Prepare-and-Measure Scenarios},
  author = {Roch i Carceller, Carles and Tavakoli, Armin},
  journal = {Phys. Rev. Lett.},
  volume = {134},
  issue = {12},
  pages = {120203},
  numpages = {6},
  year = {2025},
  month = {Mar},
  publisher = {American Physical Society},
  doi = {10.1103/PhysRevLett.134.120203},
  url = {https://link.aps.org/doi/10.1103/PhysRevLett.134.120203}
}

@article{Marton2025,
   title={Bound entanglement-assisted prepare-and-measure scenarios based on four-dimensional quantum messages},
   volume={10},
   ISSN={2058-9565},
   url={http://dx.doi.org/10.1088/2058-9565/ae095f},
   DOI={10.1088/2058-9565/ae095f},
   number={4},
   journal={Quantum Science and Technology},
   publisher={IOP Publishing},
   author={Márton, István and Bene, Erika and Vértesi, Tamás},
   year={2025},
   month=oct, pages={04LT02} }

@article{Li2021,
  title = {Activating hidden teleportation power: Theory and experiment},
  author = {Li, Jyun-Yi and Fang, Xiao-Xu and Zhang, Ting and Tabia, Gelo Noel M. and Lu, He and Liang, Yeong-Cherng},
  journal = {Phys. Rev. Res.},
  volume = {3},
  issue = {2},
  pages = {023045},
  numpages = {18},
  year = {2021},
  month = {Apr},
  publisher = {American Physical Society},
  doi = {10.1103/PhysRevResearch.3.023045},
  url = {https://link.aps.org/doi/10.1103/PhysRevResearch.3.023045}
}

@article{Pramanik2019,
  title = {Revealing hidden quantum steerability using local filtering operations},
  author = {Pramanik, Tanumoy and Cho, Young-Wook and Han, Sang-Wook and Lee, Sang-Yun and Kim, Yong-Su and Moon, Sung},
  journal = {Phys. Rev. A},
  volume = {99},
  issue = {3},
  pages = {030101},
  numpages = {5},
  year = {2019},
  month = {Mar},
  publisher = {American Physical Society},
  doi = {10.1103/PhysRevA.99.030101},
  url = {https://link.aps.org/doi/10.1103/PhysRevA.99.030101}
}

@article{Masanes2008,
  title = {All Bipartite Entangled States Display Some Hidden Nonlocality},
  author = {Masanes, Llu\'{\i}s and Liang, Yeong-Cherng and Doherty, Andrew C.},
  journal = {Phys. Rev. Lett.},
  volume = {100},
  issue = {9},
  pages = {090403},
  numpages = {4},
  year = {2008},
  month = {Mar},
  publisher = {American Physical Society},
  doi = {10.1103/PhysRevLett.100.090403},
  url = {https://link.aps.org/doi/10.1103/PhysRevLett.100.090403}
}

@article{Wang2006,
  title = {Experimental Entanglement Distillation of Two-Qubit Mixed States under Local Operations},
  author = {Wang, Zhi-Wei and Zhou, Xiang-Fa and Huang, Yun-Feng and Zhang, Yong-Sheng and Ren, Xi-Feng and Guo, Guang-Can},
  journal = {Phys. Rev. Lett.},
  volume = {96},
  issue = {22},
  pages = {220505},
  numpages = {4},
  year = {2006},
  month = {Jun},
  publisher = {American Physical Society},
  doi = {10.1103/PhysRevLett.96.220505},
  url = {https://link.aps.org/doi/10.1103/PhysRevLett.96.220505}
}

@article{Horodecki1997e,
  title = {Inseparable Two Spin- $\frac{1}{2}$ Density Matrices Can Be Distilled to a Singlet Form},
  author = {Horodecki, Micha\l{} and Horodecki, Pawe\l{} and Horodecki, Ryszard},
  journal = {Phys. Rev. Lett.},
  volume = {78},
  issue = {4},
  pages = {574--577},
  numpages = {0},
  year = {1997},
  month = {Jan},
  publisher = {American Physical Society},
  doi = {10.1103/PhysRevLett.78.574},
  url = {https://link.aps.org/doi/10.1103/PhysRevLett.78.574}
}

@article{Kwiat2001,
  title={Experimental entanglement distillation and ‘hidden’non-locality},
  author={Kwiat, Paul G and Barraza-Lopez, Salvador and Stefanov, Andre and Gisin, Nicolas},
  journal={Nature},
  volume={409},
  number={6823},
  pages={1014--1017},
  year={2001},
  publisher={Nature Publishing Group UK London}
}

@article{Ku2022,
  title={Complete classification of steerability under local filters and its relation with measurement incompatibility},
  author={Ku, Huan-Yu and Hsieh, Chung-Yun and Chen, Shin-Liang and Chen, Yueh-Nan and Budroni, Costantino},
  journal={Nature communications},
  volume={13},
  number={1},
  pages={4973},
  year={2022},
  publisher={Nature Publishing Group UK London}
}

@article{Jones2020,
  title={Exploring classical correlations in noise to recover quantum information using local filtering},
  author={Jones, Daniel E and Kirby, Brian T and Riccardi, Gabriele and Antonelli, Cristian and Brodsky, Michael},
  journal={New Journal of Physics},
  volume={22},
  number={7},
  pages={073037},
  year={2020},
  publisher={IOP Publishing}
}

@inproceedings{Zhang2023,
  title={Two-stage solution for ancilla-assisted quantum process tomography: error analysis and optimal design},
  author={Xiao, Shuixin and Wang, Yuanlong and Dong, Daoyi and Zhang, Jun},
  booktitle={2023 62nd IEEE Conference on Decision and Control (CDC)},
  pages={7178--7183},
  year={2023},
  organization={IEEE}
}

\onecolumngrid

\appendix

\section{ $3\otimes3$ PPT entangled state with maximal CCNR violation}
\label{appendix_3x3}

We now provide the explicit form of the $3\otimes3$ PPT
entangled state discussed in the main text, generated using
the numerical optimization procedure introduced in
Section~\ref{main_sec}. Since the optimization is performed
subject to the constraint $\rho^{\Gamma}\geq 0$, the
resulting state is PPT by
construction. The state is given by

\begin{equation}
\small
\rho=
\begin{pmatrix}
0.19474&0.03386&-0.00588&0.03389&-0.05209&-0.03997&0.04765&-0.02083&0.03734\\
0.03386&0.07216&0.02896&0.04847&-0.00093&-0.02711&-0.03363&-0.01904&-0.05696\\
-0.00588&0.02896&0.07508&0.00102&0.06799&-0.00988&-0.05149&-0.01154&0.00288\\
0.03389&0.04847&0.00102&0.05986&0.01951&-0.05253&0.01890&-0.02943&-0.04161\\
-0.05209&-0.00093&0.06799&0.01951&0.17277&-0.02847&0.02028&-0.07422&0.02861\\
-0.03997&-0.02711&-0.00988&-0.05253&-0.02847&0.11131&-0.01357&-0.03116&0.02362\\
0.04765&-0.03363&-0.05149&0.01890&0.02028&-0.01357&0.10703&-0.05361&0.04412\\
-0.02083&-0.01904&-0.01154&-0.02943&-0.07422&-0.03116&-0.05361&0.11615&-0.01626\\
0.03734&-0.05696&0.00288&-0.04161&0.02861&0.02362&0.04412&-0.01626&0.09090
\end{pmatrix},
\end{equation}
\normalsize

The singular values (SV) of the realigned matrix
$\mathcal{R}(\rho)$ are
\begin{align*}
    \mbox{SV}\left[\mathcal{R}(\rho)\right]
    &=
    \left\{
    0.3401,\,
    0.1712,\,
    0.1447,\,
    0.1418,\,
    0.1202,\,
    0.1197,\,
    0.0568,\,
    0.0490,\,
    0.0455
    \right\}.
\end{align*}

Since all singular values of
$\mathcal{R}(\rho)$ are strictly positive, the realigned
matrix is invertible. Consequently, $\rho$ is faithful and
can therefore serve as a probe state for AAQPT.
Furthermore,

\begin{equation}
\left\|
\mathcal{R}(\rho)
\right\|_{\mathrm{tr}}
=
1.1890,
\end{equation}

which certifies entanglement via the CCNR criterion.
All nine singular values are strictly positive, so the state has full
operator Schmidt rank and is therefore faithful.

\end{document}